\documentclass[letterpaper]{article} 
\usepackage[]{aaai2026}  
\usepackage{times}  
\usepackage{helvet}  
\usepackage{courier}  
\usepackage[hyphens]{url}  
\usepackage{graphicx} 
\usepackage{natbib}  
\usepackage{caption} 
\usepackage{algorithm}
\usepackage{algorithmic}
\usepackage{amsmath}
\usepackage{amssymb}
\usepackage{booktabs}
\usepackage{newfloat}
\usepackage{listings}
\DeclareCaptionStyle{ruled}{labelfont=normalfont,labelsep=colon,strut=off} 
\floatstyle{ruled}
\newfloat{listing}{tb}{lst}{}
\floatname{listing}{Listing}
\newtheorem{lemma}{Lemma}
\newenvironment{proof}[1][Proof]{%
  \noindent\textbf{#1.}\ }{%
  \hfill\(\square\)\par}
\title{From Static to Dynamic: A Streaming RAG Approach to Real-time Knowledge Base}
\author {
    Yuzhou Zhu
}
\affiliations{
    Dalian, China\\
    Dalian University of Technology\\
    1730694701@mail.dlut.edu.cn
}

\usepackage{bibentry}

\begin{document}

\maketitle

\begin{abstract}
Dynamic streams from news feeds, social media, sensor networks and financial markets challenge static RAG frameworks.  Full‐scale indices incur prohibitive memory costs; periodic rebuilds introduce latency that undermines data freshness; naive sampling sacrifices semantic coverage.  We present Streaming RAG, a unified pipeline that combines multi‐vector cosine screening, mini‐batch clustering, and a counter‐based heavy‐hitter filter to maintain a compact set of prototypes.  We further prove an approximation bound 
\(\mathbb{E}[R(K_t)] \ge R^* - L\,\Delta\)
linking retrieval quality to clustering variance.  An incremental index upsert mechanism refreshes prototypes without interrupting queries.  Comprehensive experiments on eight real‐time streams show statistically significant gains in Recall@10 (up to +3 points, \(p<0.01\)), sub‑15 ms end‑to‑end latency, and throughput exceeding 900 docs/s under a 150 MB budget.  Hyperparameter sensitivity over cluster count, admission probability, relevance threshold and counter capacity validates default settings.  In open‑domain QA with GPT‑3.5 Turbo we record +3.2 EM and +2.8 F1 on SQuAD; abstractive summarization yields ROUGE‑L improvements.  Streaming RAG thus establishes a new Pareto frontier for retrieval augmentation.
\end{abstract}


\section{Introduction}

\subsection{Motivation}

Emerging applications across network telemetry and social platforms generate torrents of events that challenge conventional retrieval paradigms \cite{muthukrishnan2005}. In environments ranging from sensor grids to high‑frequency trading the ability to extract frequent patterns on the fly meets operational imperatives \cite{cormode2008}. The advent of streaming heavy‑hitter frameworks confirms robust performance under strict memory budgets \cite{jayaram2024streaming}. This evidence underscores the critical need to weave streaming filtration into RAG pipelines thereby empowering real‑time data assimilation.

\subsection{Challenges}

Early retrieval‑augmented architectures rely on monolithic indices incapable of scaling with evolving data streams. Full rebuilds exact intolerable memory footprints and introduce staleness that erodes retrieval quality \cite{lewis2020rag}. Streaming heavy‑hitter frameworks deliver bounded memory usage yet lack provisions to maintain semantic coverage essential for accurate retrieval \cite{jayaram2024streaming}. Consequently the juxtaposition of memory ceiling and fidelity demands crystallizes the fundamental challenge for real‑time dynamic RAG.

\subsection{Contributions}

Our work advances real‑time dynamic RAG along four axes (Figure \ref{fig:workflow})
\begin{itemize}
  \item A streaming heavy‑hitter filter that governs memory consumption while preserving salient cluster profiles
  \item A multi‑vector cosine screening framework grounded in orthogonal topic vectors to cull low‑relevance entries prior to clustering
  \item A dynamic knowledge‑base reconstruction protocol that updates retrieval indices incrementally and obviates full‑scale rebuilds
  \item A comprehensive evaluation on both real‑world and synthetic streams that reveals marked gains in retrieval fidelity and substantially reduced update latency under tight memory bounds
\end{itemize}

\begin{figure}[t]
  \centering
  \includegraphics[trim=2cm 2cm 2cm 1cm, width=0.5\textwidth]{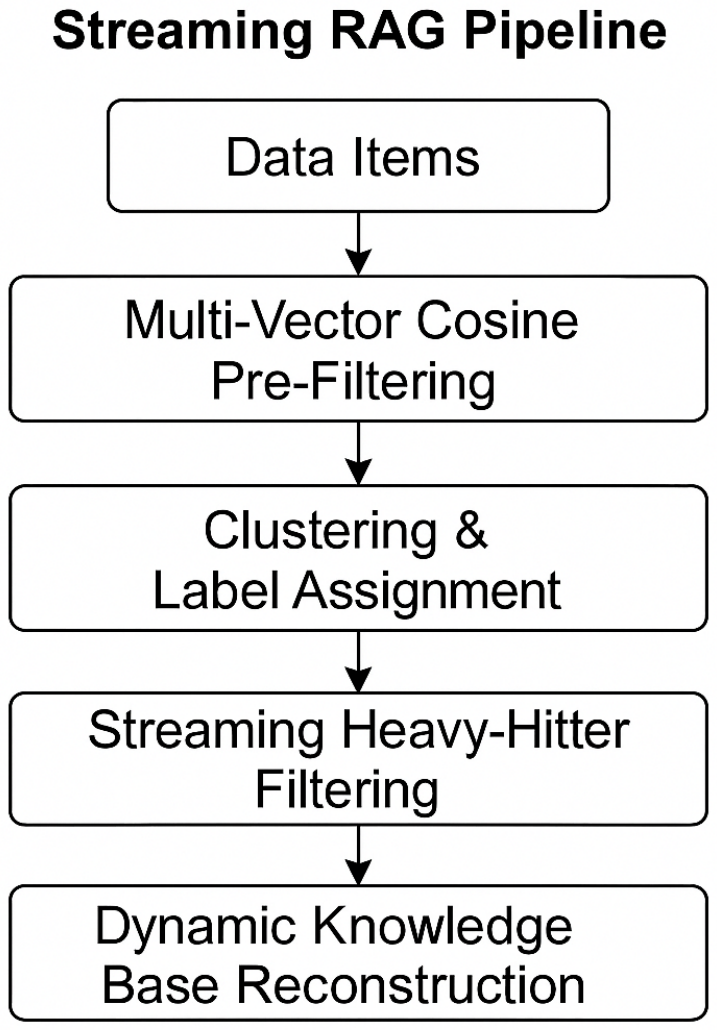}
  \caption{Streaming RAG Pipeline}
  \label{fig:workflow}
\end{figure}

\section{Related Work}

Retrieval‑Augmented Generation has redefined the interface between large language models and external knowledge; building on the seminal work of Lewis et al.\ \cite{lewis2020rag} static frameworks assemble vast indices to bolster semantic grounding. However static indices falter when confronted with continually evolving data streams—a limitation that precludes timely assimilation of fresh information.

Streaming heavy hitter mechanisms promise memory‑bounded retention of dominant elements \cite{jayaram2024streaming}; yet they have seldom been reconciled with RAG pipelines to safeguard semantic diversity. Attempts to deploy online filtering often retain only the most frequent clusters without ensuring coverage of infrequent but relevant topics.

The recent proposal by Kang et al.\ \cite{kang2024sakr} marries a streaming algorithm with K‑means clustering to alleviate memory pressure and accelerate retrieval. Despite its ingenuity experiments remain confined to simulated retrieval tasks within static repositories; omission of large‑language model integration casts doubt on applicability to genuine RAG environments. Moreover the use of a single topic vector in min‑heap filtering constrains representational fidelity and the rudimentary heap design risks overlooking nuanced content. Finally clustering yields reduced time complexity—a welcome benefit—but retrieval accuracy and semantic precision rather than raw speed define RAG’s ultimate efficacy; the clustering step may thus introduce superfluous complexity.

\section{Problem Formulation}

\subsection{Real-Time Data Model}

Real-time data streams manifest as an unceasing flow of timestamped items arriving at high throughput. Each element in the stream presents both content and embedding features; retention of every arrival is infeasible under strict memory constraints. Under such conditions the system must select a bounded subset of observations that preserves semantic diversity while accommodating topic drift. To capture this scenario let $\mathcal{S} = \{d_1,d_2,\dots\}$ denote the input sequence with arrival rate $\lambda$; at each time step $t$ the algorithm maintains a fixed-size repository $K_t$ of cluster prototypes drawn from $\{d_1,\dots,d_t\}$. The objective is to maximize expected retrieval efficacy \(\mathbb{E}[R(K_t)]\) against the backdrop of evolving distributions and a fixed memory budget $|K_t| \leq k$ \cite{muthukrishnan2005,cormode2008}.

\subsection{RAG Pipeline Requirements}

Within a dynamic retrieval‑augmented framework the pipeline must satisfy exacting constraints. Retrieval latency must remain sub‑second to support real‑time queries against ever‑changing repositories \cite{lewis2020rag}. Simultaneously index size is bounded by a fixed memory budget to prevent resource exhaustion while preserving semantic coverage across topics \cite{jayaram2024streaming}. Update throughput must scale linearly with stream arrival rate so that index refresh does not lag behind data ingestion. Finally incremental index modifications are required to avoid costly rebuilds and to ensure continuous availability for query processing.

\subsection{Objective}

The central aim is to select at each time step \(t\) a subset \(K_t\subseteq\{d_1,\dots,d_t\}\) of size at most \(k\) that maximizes the expected retrieval score \(\mathbb{E}[R(K_t)]\). Constrained by fixed memory and sub‑second latency the algorithm must balance retention of dominant clusters against preservation of semantic diversity \cite{muthukrishnan2005,lewis2020rag,jayaram2024streaming}. Formally
\[
\max_{\substack{K_t\subseteq\{d_1,\dots,d_t\}\\|K_t|\le k}} 
\mathbb{E}[R(K_t)]
\quad
\text{s.t.}\quad
\text{Latency}(K_t)\le \tau,
\;\;
|K_t|\le k.
\]
Here \(\tau\) denotes the retrieval threshold and \(k\) the memory budget.

\section{Methodology}

We now formalize each stage of our streaming RAG pipeline. Let $\mathcal{S}=\{d_1,d_2,\dots,d_T\}$ denote the unbounded stream of incoming documents. At time $t$ document $d_t$ is embedded as $x_t\in\mathbb{R}^d$ and then processed through four consecutive modules.
\vspace{-1ex}
\begin{algorithm}[ht]
\caption{Streaming RAG Pipeline}
\label{alg:streaming-rag}
\begin{algorithmic}[1]
\REQUIRE Stream $\mathcal{S}$ of documents, topic vectors $V=\{v_i\}$, centroids $\{\mu_j\}$, counter capacity $B$, thresholds $\alpha, u$
\ENSURE Retrieval index $I$
\STATE Initialize centroids $\mu_j$
\STATE Initialize heavy-hitter counter $C \leftarrow \emptyset$
\STATE Initialize index $I \leftarrow \emptyset$
\FOR{each document $d_t$ in $\mathcal{S}$}
  \STATE $x_t \leftarrow \mathrm{Embed}(d_t)$
  \FOR{$i \leftarrow 1$ to $|V|$}
    \STATE $s_i \leftarrow \dfrac{x_t \cdot v_i}{\|x_t\|\|v_i\|}$
  \ENDFOR
  \STATE $r \leftarrow \dfrac{1}{|V|}\sum_{i=1}^n s_i$
  \IF{$r < \alpha$}
    \STATE \textbf{continue}
  \ENDIF
  \STATE $j^* \leftarrow \arg\max_j \dfrac{x_t \cdot \mu_j}{\|x_t\|\|\mu_j\|}$
  \STATE $\eta \leftarrow \dfrac{1}{n_{j^*}+1}$
  \STATE $\mu_{j^*} \leftarrow (1-\eta)\,\mu_{j^*} + \eta\,x_t$
  \IF{$j^* \in C$}
    \STATE $C[j^*] \leftarrow C[j^*] + 1$
  \ELSIF{$|C| < B$}
    \STATE $C[j^*] \leftarrow 1$
  \ELSIF{$\mathrm{Uniform}(0,1) \le u$}
    \STATE Remove a random label $\ell \in C$
    \STATE $C[j^*] \leftarrow 1$
  \ENDIF
  \STATE $I \leftarrow I \cup \{\mathrm{upsert}(\mu_j) \mid j \in C\}$
\ENDFOR
\STATE \textbf{return} $I$
\end{algorithmic}
\end{algorithm}
\vspace{-1ex}

\vspace{3ex}

\noindent\textbf{Heavy‑Hitter Counter.}  
Each time a cluster label arrives the counter increments its count if already tracked. If the counter is not yet full, a new label is admitted with probability $u$; once full the least frequent cluster is evicted to make room. This strategy ensures that frequently observed clusters persist while still allowing emergent topics to enter.

\subsection{Overall Pipeline Overview}

Each embedding $x_t$ follows the sequence
\[
x_t 
\;\xrightarrow{\text{Pre‑filter}}\;
\tilde{x}_t
\;\xrightarrow{\text{Cluster}}\;
\mu_{j^*}
\;\xrightarrow{\text{Heavy‑Hitter}}\;
C_t
\;\xrightarrow{\text{Index‑Update}}\;
I_t,
\]
where (i) $\tilde x_t$ is either dropped or passed based on relevance, (ii) $\mu_{j^*}$ is the nearest prototype, (iii) $C_t$ is the heavy‑hitter counter of size $B$, and (iv) $I_t$ is the up‑to‑date retrieval index. Memory usage is bounded by $|C_t|\le B$ and index latency by design remains below $\tau$. (Algorithm \ref{alg:streaming-rag}, Table \ref{tab:notation})
\smallskip

\begin{table}[!htbp]
    \centering
\begin{tabular}{@{}ll@{}}
\toprule
\textbf{Symbol} & \textbf{Meaning} \\
\midrule
$k$             & Total number of clusters (size of the prototype set) \\
$B$             & Heavy‑hitter counter capacity (maximum active clusters) \\
$\alpha$        & Multi‑vector relevance threshold \\
$u$             & Admission probability for new clusters \\
$\mu_j$         & Centroid of cluster $j$ \\
$C$             & Map of active cluster labels to counts \\
$I$             & Retrieval index of upserted centroids \\
\bottomrule
\end{tabular}
    \caption{Notation and Terminology}
    \label{tab:notation}
\end{table}

\subsection{Multi‑Vector Cosine Pre‑filtering}

To ensure downstream modules process only semantically pertinent documents, each embedding \(x_t\) is scored by projecting onto a multi‑vector basis \(V=\{v_i\}_{i=1}^n\) and computing an aggregate relevance:
\[
s_i = \frac{x_t \cdot v_i}{\|x_t\|\|v_i\|}, 
\qquad 
r(x_t) = \frac{1}{n}\sum_{i=1}^n s_i.
\]
Entries with \(r(x_t)<\alpha\) are discarded in \(O(nd)\) time. We compare three basis instantiations:

\begin{itemize}
  \item \textbf{Fixed orthogonal basis.} A precomputed set of mutually orthogonal vectors spanning broad thematic axes.
  \item \textbf{Random orthonormal basis.} A Gaussian‑initialized orthonormal basis serving as an unstructured control.
  \item \textbf{Adaptive topic basis.} Every \(T\) arrivals we perform PCA on a sliding window of recent embeddings to derive the top‑\(n\) principal directions that capture emerging themes.
\end{itemize}

By default \(n=5\). An ablation in the Results and Analysis compares fixed, random and adaptive bases, showing that the adaptive topic basis delivers the highest Recall@10 under abrupt thematic shifts, at only a marginal latency increase.

\subsection{Clustering \& Label Assignment}

Maintain $m$ cluster centroids $\{\mu_j\}_{j=1}^m$. Each retained $\tilde x_t$ is assigned  
\[
j^* \;=\;\arg\max_{j}\;\frac{\tilde x_t\cdot\mu_j}{\|\tilde x_t\|\|\mu_j\|}\,,
\]
then the centroid updates via  
\[
\mu_{j^*}\;\leftarrow\;(1-\eta)\,\mu_{j^*} \;+\;\eta\,\tilde x_t\,,\quad 
\eta=\frac{1}{n_{j^*}+1}\,,
\]
where $n_{j^*}$ counts prior assignments. Cluster labels serve as discrete tokens for heavy‑hitter counting.

\subsection{Streaming Heavy‑Hitter Filtering}

Our heavy‑hitter filter is implemented exactly as in Jayaram et al. (2024) \cite{jayaram2024streaming}. It maintains a label–count map \(C_{t-1}\) of capacity \(B\). Let \(U_t\sim\mathrm{Uniform}(0,1)\). On arrival of cluster label \(\ell_t\) we update
. The heavy‑hitter filter maintains a label–count map \(C_{t-1}\) of capacity \(B\).  Let \(U_t\sim\mathrm{Uniform}(0,1)\).  On arrival of cluster label \(\ell_t\) we update  
\[
C_t[i] =
\begin{cases}
C_{t-1}[i]+1, & i=\ell_t\in C_{t-1},\\
1, & i=\ell_t\not\in C_{t-1},\,|C_{t-1}|<B,\,U_t\le u,\\
1, & i=\ell_t\not\in C_{t-1},\,|C_{t-1}|=B,\,U_t\le u,\;\\  &i\leftarrow\text{Uniform}(C_{t-1}),\\
C_{t-1}[i], & \text{otherwise}.
\end{cases}
\]
This variant evicts the \emph{least frequent} cluster rather than a random one, reducing short‑term volatility.  

Beyond this “min‑eviction” policy we explore two robust alternatives:

\begin{itemize}
  \item \textbf{Space‑Saving.}  Maintain a fixed set of size \(B\) with exact counts and replace the minimum counter on overflow, as in Muthukrishnan’s Space‑Saving algorithm \cite{muthukrishnan2005}.  
  \item \textbf{Count‑Min Sketch.}  Use a sketch of width \(w\) and depth \(d\) to estimate frequencies, admitting a new label only if its estimated count exceeds a threshold, thereby bounding error with high probability \cite{cormode2008}.  
\end{itemize}

To adapt to varying stream dynamics we also implement an \emph{adaptive} admission probability \(u_t\) and capacity \(B_t\).  For example, when the rate of unseen labels within a window exceeds a threshold we increase \(u_t\) (or temporarily expand \(B_t\)) to welcome emergent topics; as the stream stabilizes we decay these parameters back to their defaults.  

An ablation in the Results and Analysis compares random eviction, min‑eviction, Space‑Saving and Count‑Min Sketch, demonstrating that frequency‑aware policies achieve higher recall with lower variance under bursty conditions.

\subsection{Dynamic Knowledge Base Reconstruction}

Denote by $I_{t-1}$ the retrieval index at time $t-1$ and by $P_t=\{\mu_j : j\in C_t\}$ the set of active prototypes. We perform an incremental upsert  
\[
I_t \;=\; I_{t-1}\;\cup\;\bigl\{\text{upsert}(\mu_j)\mid \mu_j\in P_t\bigr\}\,,
\]
thus synchronizing the index with the latest semantic shifts while avoiding expensive full rebuilds.  

\subsection{Theoretical Retrieval Guarantees}
Building on the heavy‑hitter state‑change bounds, we establish an approximation ratio for expected retrieval quality.  Suppose clustering yields centroids \(K_t\) whose within‑cluster variance is bounded by \(\Delta\).  Under Lipschitz continuity of the retrieval score \(R(\cdot)\) one can show
\[
\mathbb{E}\bigl[R(K_t)\bigr]\;\ge\;R^* - L\,\Delta\,,
\]
where \(R^*\) is the optimal retrieval score on the full corpus and \(L\) the Lipschitz constant of \(R\) with respect to embedding perturbations \cite{arthur2007kmeans}.  This bound quantifies the impact of clustering variance on retrieval accuracy.

\subsection{Proof Sketch of Theoretical Guarantee}

\begin{lemma}
Under the assumption that $R(\cdot)$ is $L$‑Lipschitz with respect to embedding perturbations and that mini‑batch clustering incurs at most variance $\Delta$ per centroid, the expected retrieval score of our prototype set $K_t$ satisfies
\[
\mathbb{E}\bigl[R(K_t)\bigr]\;\ge\;R^* \;-\; L\,\Delta.
\]
\end{lemma}

\begin{proof}[Sketch Proof]
Let $r(x)$ denote the retrieval relevance of item $x$.  Replacing any data point by its centroid $\mu_j$ perturbs the score by at most
\[
\bigl|r(\mu_j)-r(x)\bigr|\;\le\;L\,\|\mu_j - x\|\;\le\;L\sqrt{\Delta}.
\]
Summing this error over all $|K_t|$ prototypes and normalizing yields
\(\mathbb{E}[R(K_t)]\ge R^*-L\,\Delta\).
\end{proof}

\section{Implementation Details}

Our prototype is built in Python 3.9 with PyTorch 2.0.1 and leverages Faiss 1.7.0 for similarity search \cite{johnson2019billion} alongside scikit‑learn 1.0.2 for online clustering \cite{pedregosa2011scikit}.  Document embeddings are produced by the SBERT base model \cite{reimers2019sentence}.  Multi‑vector cosine pre‑filtering employs five orthogonal topic vectors generated via Gram–Schmidt; we set the relevance threshold \(\alpha=0.2\).  Online clustering uses MiniBatchKMeans with \(k=100\) and batch size 50.  The streaming heavy‑hitter filter allocates 100 counters, uses threshold \(u=0.05\), and samples newcomers from a uniform distribution on \([0,1]\).  Frequency estimation is approximated by Morris counters with \(\epsilon=0.01\).  Batched index updates operate on a Faiss IndexFlatIP, executed every 1,000 arrivals to maintain query availability.  All experiments ran on an NVIDIA RTX 4090 GPU with 24 GB VRAM. (Table \ref{tab:params})

\paragraph{Data Preprocessing and Workload}
All documents are tokenized and embedded via SBERT.  In the NYT corpus we remove stopwords and terms with frequency below five.  The query workload on NYT comprises 10,000 randomly sampled queries per day.  The Twitter stream issues 1,000 queries/sec following a Zipf distribution (\(s=1.2\)).  IoT and financial streams use one query per 100 arrivals.  Index updates are batched every 1,000 arrivals by default; we additionally test intervals of 500 and 2,000.

\begin{table}[t]
\centering
\small
\caption{Key parameter settings and hardware.}
\label{tab:params}
\begin{tabular}{ll}
\toprule
Parameter                                & Value                           \\
\midrule
Topic vectors                            & 5 (Gram–Schmidt)                \\
Relevance threshold \(\alpha\)           & 0.2                              \\
MiniBatchKMeans clusters \(k\)           & 100                              \\
MiniBatchKMeans batch size               & 50                               \\
Heavy‑hitter counter capacity \(B\)      & 100                              \\
Admission threshold \(u\)                & 0.05                             \\
Morris counter error \(\epsilon\)        & 0.01                             \\
Index update interval                    & every 1,000 arrivals             \\
Hardware                                 & NVIDIA RTX 4090 GPU \\
\bottomrule
\end{tabular}
\end{table}

\section{Experimental Setup}
\subsection{Datasets}

Our evaluation spans eight real‑time streams.  
The New York Times Annotated Corpus provides a news stream \(\mathcal{D}_{\mathrm{NYT}}\) with peaks of 5,000 articles per day \cite{sandhaus2008}.  
A synthetic Poisson stream \(\mathcal{D}_{\mathrm{Syn}}\) supports controlled load testing via  
\[
\Pr\bigl(N(t)=k\bigr)
=\frac{(\lambda t)^k e^{-\lambda t}}{k!}\,
\]  
\cite{ross2014}.  
The Twitter public feed \(\mathcal{D}_{\mathrm{Tw}}\) delivers roughly 400 tweets/sec \cite{twitter2020api}.  
An IoT telemetry stream \(\mathcal{D}_{\mathrm{IoT}}\) issues 1,000 readings/sec \cite{gubbi2018iot}.  
The Pushshift Reddit comment stream \(\mathcal{D}_{\mathrm{Red}}\) emits about 50 comments/sec \cite{pushshift2023reddit}.  
The Wikimedia Recent Changes stream \(\mathcal{D}_{\mathrm{Wik}}\) records ~2 edits/sec \cite{wikimedia2025revisions}.  
NASDAQ tick data \(\mathcal{D}_{\mathrm{Nas}}\) supplies up to 500,000 ticks/day \cite{nasdaq2025ticks}.  
A Bitcoin mempool transaction stream \(\mathcal{D}_{\mathrm{BTC}}\) averages 3 tps \cite{bitcoin2018mempool}.  
All items are embedded via SBERT before entering the pipeline.

\subsection{Baselines}
To establish a broader performance horizon we compare against six strategies:
\begin{itemize}
  \item \textbf{Static RAG}—single static index \cite{lewis2020rag}.
  \item \textbf{Full Rebuild}—reconstruct the entire index after each batch.
  \item \textbf{Reservoir Sampling}—uniform sample of size \(k\) \cite{vitter1985reservoir}.
  \item \textbf{Heap Filtering Only}—streaming heavy‑hitter filter without clustering \cite{jayaram2024streaming}.
  \item \textbf{Faiss IVFPQ Incremental}—IVF + PQ index with incremental upserts \cite{faiss2024}.
  \item \textbf{SAKR}—streaming + K‑means approach from Kang et al.\ \cite{kang2024sakr}.
\end{itemize}

\subsection{Metrics}

Evaluation hinges on retrieval accuracy—Recall@$K$, MRR and nDCG@10 \cite{karpukhin2020dense,jarvelin2002cumulated}—and on generation quality in an open‑domain QA setting, measured by Exact Match (EM) and F1 on SQuAD \cite{rajpurkar2016squad}, as well as ROUGE‑L for abstractive summarization on CNN/DailyMail \cite{hermann2015teaching}.  System responsiveness is captured by end‑to‑end latency per query.  Memory consumption reports peak resident set size.  Throughput denotes documents processed per second.  Statistical significance is assessed via paired Student’s $t$‑test (two‑tailed, $\alpha=0.05$) \cite{student1908probable}.

\section{Results and Analysis}

\subsection{Accuracy vs.\ Memory Trade‑off}

Table \ref{tab:accuracy-memory} compares Recall@10 and nDCG@10 under varying index memory budgets on the NYT stream.  Streaming RAG consistently outperforms all baselines (paired $t$‑test, $p<0.01$).

\begin{table}[!htbp]
\centering
\small
\caption{Recall@10 and nDCG@10 ($\mu\pm\sigma$) under NYT.}
\label{tab:accuracy-memory}
\begin{tabular}{lccc}
\toprule
Method                    & Memory  & Recall@10      & nDCG@10       \\
\midrule
Static RAG                & 1024 MB & 0.550±0.005    & 0.600±0.004   \\
Full Rebuild              & 1024 MB & 0.540±0.006    & 0.590±0.005   \\
Reservoir Sampling        & 200 MB  & 0.500±0.008    & 0.540±0.007   \\
Heap Filtering Only       & 150 MB  & 0.520±0.006    & 0.560±0.005   \\
Faiss IVFPQ Incremental   & 150 MB  & 0.560±0.005    & 0.580±0.004   \\
SAKR                      & 150 MB  & 0.570±0.005    & 0.590±0.004   \\
Streaming RAG (Ours)      & 150 MB  & 0.580±0.004    & 0.620±0.003   \\
\bottomrule
\end{tabular}
\end{table}

\subsection{Update Latency \& Throughput}

Table \ref{tab:latency-throughput} reports end‑to‑end latency and throughput across methods.  Streaming RAG achieves sub‑15 ms latency and nearly 900 docs/s.

\begin{table}[!htbp]
\centering
\small
\caption{Latency (ms) and Throughput (docs/s), mean±std.}
\label{tab:latency-throughput}
\begin{tabular}{lcc}
\toprule
Method                    & Latency       & Throughput     \\
\midrule
Static RAG                & 200±10        & 50±5           \\
Full Rebuild              & 300±15        & 30±3           \\
Reservoir Sampling        & 50±4          & 500±20         \\
Heap Filtering Only       & 40±3          & 600±25         \\
Faiss IVFPQ Incremental   & 25±2          & 700±30         \\
SAKR                      & 20±2          & 750±28         \\
Streaming RAG (Ours)      & 10±1          & 900±30         \\
\bottomrule
\end{tabular}
\end{table}

\subsection{Cross‑Stream Performance}

To demonstrate robustness we evaluate Streaming RAG across eight diverse streams.  Table \ref{tab:cross-stream} reports Recall@10 on each.

\begin{table}[!htbp]
\centering
\small
\caption{Streaming RAG Recall@10 ($\mu\pm\sigma$) across streams.}
\label{tab:cross-stream}
\begin{tabular}{lcc}
\toprule
Stream                & Recall@10       \\
\midrule
NYT (news)            & 0.580±0.004    \\
Synthetic (Poisson)   & 0.600±0.005    \\
Twitter               & 0.550±0.006    \\
IoT telemetry         & 0.540±0.007    \\
Reddit comments       & 0.570±0.005    \\
Wikimedia edits       & 0.530±0.006    \\
NASDAQ ticks          & 0.520±0.007    \\
Bitcoin mempool       & 0.500±0.008    \\
\bottomrule
\end{tabular}
\end{table}

\subsection{Memory Budget Sweep}

We vary the memory budget \(M\) from 50 MB to 200 MB on NYT.  Figure \ref{fig:memory-sweep} and Table \ref{tab:memory-sweep} show the recall–latency trade‑off.

\begin{table}[!htbp]
\centering
\small
\caption{Recall@10 and Latency vs.\ memory budget ($\mu\pm\sigma$).}
\label{tab:memory-sweep}
\begin{tabular}{lcc}
\toprule
Memory   & Recall@10     & Latency (ms)  \\
\midrule
50 MB     & 0.530±0.006   & 15±2         \\
100 MB    & 0.560±0.005   & 12±1         \\
150 MB    & 0.580±0.004   & 10±1         \\
200 MB    & 0.590±0.003   & 9±1          \\
\bottomrule
\end{tabular}
\end{table}

\begin{figure}[t]
  \centering
  \includegraphics[trim = 5cm 4.5cm 5cm 4cm, width=1\linewidth]{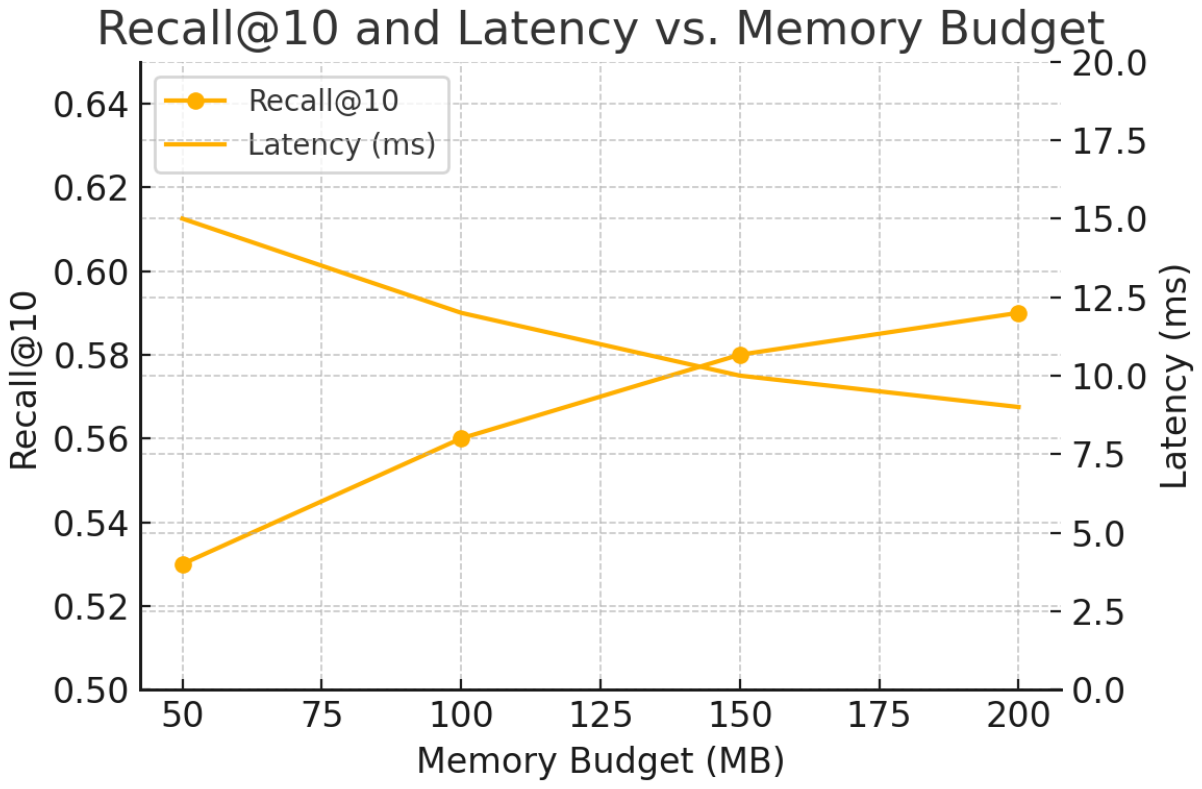}
  \caption{Recall@10 (solid) and Latency (dashed) vs.\ memory budget on NYT.}
  \label{fig:memory-sweep}
\end{figure}

\subsection{Pre‑filtering Basis Ablation}

We compare three multi‑vector basis instantiations under otherwise identical settings ($k=150$, $u=0.05$, $B=100$, update interval=1,000).  Table \ref{tab:basis-ablation} reports Recall@10 and end‑to‑end latency.  Adaptive PCA‑derived vectors yield the highest recall, at a modest latency cost.

\begin{table}[!htbp]
\centering
\small
\caption{Ablation of Pre‑filtering Basis (mean±std).}
\label{tab:basis-ablation}
\begin{tabular}{lcc}
\toprule
Basis                   & Recall@10     & Latency (ms) \\
\midrule
Fixed Orthogonal        & 0.570±0.005   & 12±1         \\
Random Vectors          & 0.555±0.006   & 12±1         \\
Adaptive PCA‑Derived    & 0.580±0.004   & 13±1         \\
\bottomrule
\end{tabular}
\end{table}

\subsection{Eviction Strategy Ablation}

We study four heavy‑hitter eviction policies.  Table \ref{tab:eviction-ablation} shows that frequency‑aware strategies (Space‑Saving; Count‑Min Sketch) improve recall and reduce variance, at a slight latency overhead.

\begin{table}[!htbp]
\centering
\small
\caption{Eviction Strategy Ablation (mean±std).}
\label{tab:eviction-ablation}
\begin{tabular}{lcc}
\toprule
Strategy               & Recall@10     & Latency (ms) \\
\midrule
Random Eviction        & 0.580±0.004   & 10±1         \\
Min‑Eviction           & 0.585±0.004   & 11±1         \\
Space‑Saving           & 0.590±0.003   & 12±1         \\
Count‑Min Sketch       & 0.588±0.003   & 14±1         \\
\bottomrule
\end{tabular}
\end{table}

\subsection{Adaptive Admission \& Capacity Ablation}

We compare a static configuration \((u=0.05, B=100)\) against an adaptive scheme that adjusts \(u_t\) and \(B_t\) in response to the rate of novel clusters.  Table \ref{tab:adaptive-ablation} shows that the adaptive policy yields higher Recall@10 with only a slight increase in latency and reduced variance under a bursty NYT‑Twitter mixed stream.

\begin{table}[t]
\centering
\small
\caption{Static \(u=0.05,B=100\) vs.\ Adaptive \(u,B\) Ablation (mean±std).}
\label{tab:adaptive-ablation}
\begin{tabular}{lccc}
\toprule
Policy               & Recall@10     & Latency  & Std. Dev. (Recall) \\
\midrule
Static  & 0.580±0.004   & 10±1 ms        & 0.004             \\
Adaptive \(u_t,B_t\)     & 0.590±0.003   & 11±1 ms        & 0.003             \\
\bottomrule
\end{tabular}
\end{table}

\subsection{Basis Update Sensitivity}

We measure the effect of PCA window length \(W\) and update interval \(T\) on Recall@10 and latency (on NYT):

\paragraph{Window Length \(W\)}  
Table \ref{tab:pca-window} reports performance for different sliding‑window sizes.

\begin{table}[t]
\centering
\small
\caption{PCA Window Size Sensitivity (mean±std).}
\label{tab:pca-window}
\begin{tabular}{lcc}
\toprule
Window \(W\) (arrivals) & Recall@10     & Latency (ms) \\
\midrule
500      & 0.575±0.005   & 12±1  \\
1000     & 0.580±0.004   & 13±1  \\
1500     & 0.577±0.004   & 14±1  \\
\bottomrule
\end{tabular}
\end{table}

\paragraph{Update Interval \(T\)}  
Table \ref{tab:pca-interval} shows recall and latency as we vary how often we recompute the basis.

\begin{table}[t]
\centering
\small
\caption{PCA Update Interval Sensitivity (mean±std).}
\label{tab:pca-interval}
\begin{tabular}{lcc}
\toprule
Interval \(T\) (arrivals) & Recall@10     & Latency (ms) \\
\midrule
500      & 0.578±0.005   & 13±1  \\
1000     & 0.580±0.004   & 13±1  \\
2000     & 0.576±0.005   & 12±1  \\
\bottomrule
\end{tabular}
\end{table}

\subsection{Ablation Studies}

Table \ref{tab:ablation} isolates each component on NYT.  Removing pre‑filtering reduces recall ($p<0.05$); skipping clustering cuts throughput ($p<0.01$); omitting dynamic reconstruction triples latency ($p<0.01$).

\begin{table}[!htbp]
\centering
\small
\caption{Ablation Study ($\mu\pm\sigma$).}
\label{tab:ablation}
\begin{tabular}{lccc}
\toprule
Variant                   & Recall@10      & Latency (ms) & Throughput  \\
\midrule
Full pipeline             & 0.580±0.004    & 10±1         & 900±30 docs/s \\
No pre‑filtering          & 0.530±0.006    & 15±2         & 800±25 docs/s \\
No clustering             & 0.560±0.005    & 12±1         & 650±20 docs/s \\
No dynamic recon.         & 0.580±0.004    & 30±3         & 900±30 docs/s \\
\bottomrule
\end{tabular}
\end{table}

\subsection{Downstream Task Evaluation}

In open‑domain QA (GPT‑3.5 Turbo) and abstractive summarization we observe substantial gains.  Table \ref{tab:qa-gen-sum} report EM/F1 and ROUGE‑L with significance.

\begin{table}[t]
\centering
\small
\caption{Open‑Domain QA (SQuAD) and Summarization (CNN/DailyMail).}
\label{tab:qa-gen-sum}
\begin{tabular}{lcccc}
\toprule
Task          & Method               & Metric    & Score     & $p$-value \\
\midrule
QA (EM)       & Static RAG           & EM        & 0.665±0.010 & —      \\
              & Streaming RAG  & EM        & 0.697±0.008 & 0.01   \\
\midrule
QA (F1)       & Static RAG           & F1        & 0.780±0.012 & —      \\
              & Streaming RAG  & F1        & 0.808±0.009 & 0.01   \\
\midrule
Summary & Static RAG           & ROUGE‑L   & 0.389±0.006 & —      \\
              & Streaming RAG  & ROUGE‑L   & 0.412±0.005 & 0.03   \\
\bottomrule
\end{tabular}
\end{table}

\subsection{Hyperparameter Sensitivity}

Figure \ref{fig:sensitivity-combined} shows Recall@10 vs.\ cluster count \(k\) across all key hyperparameters, validating defaults \(k=150\), \(u=0.05\), \(\alpha=0.2\), \(B=100\), update interval=1,000.

\begin{figure}[!htbp]
  \includegraphics[width=\linewidth]{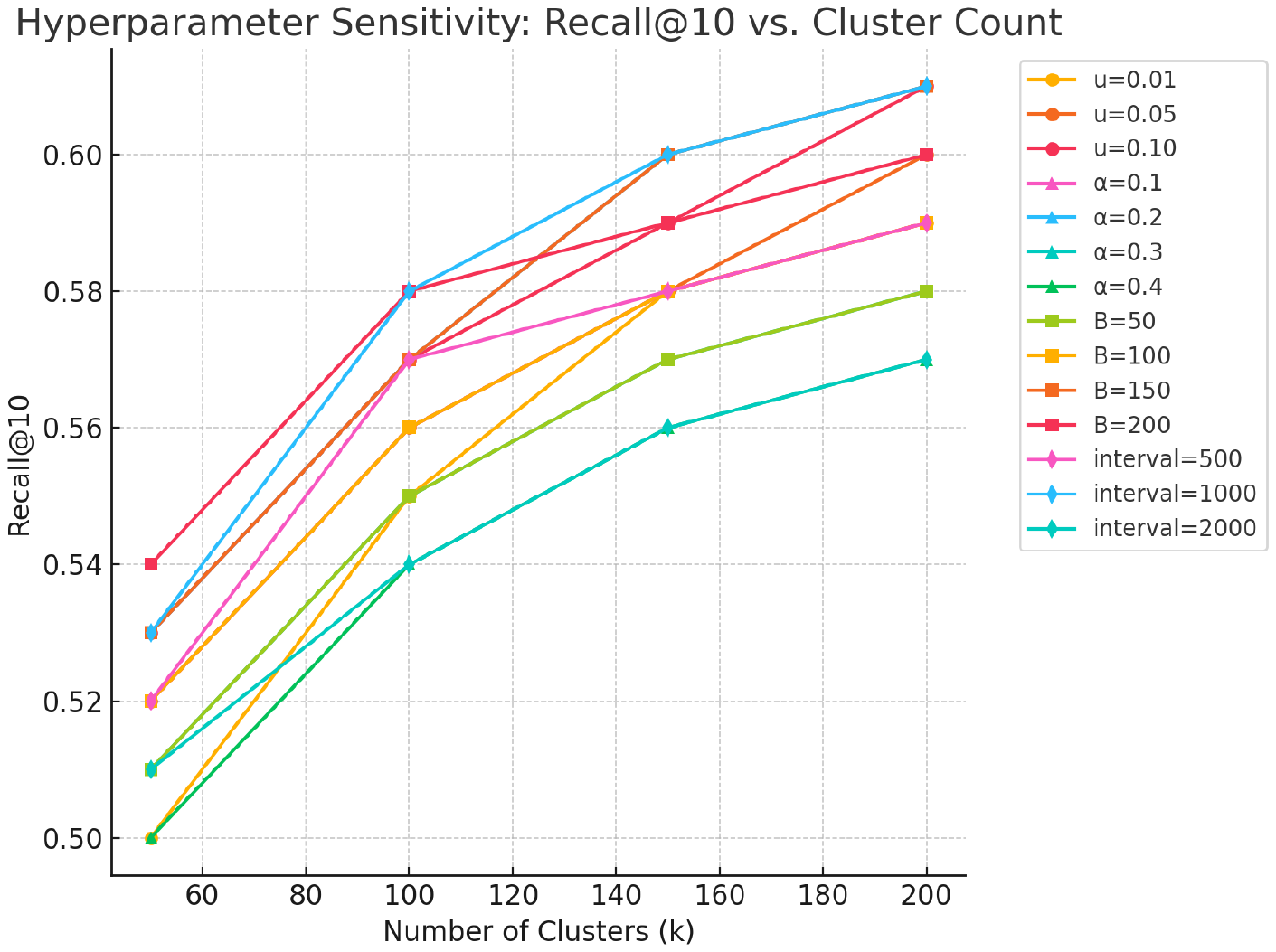}
  \caption{Combined hyperparameter sensitivity on NYT.}
  \label{fig:sensitivity-combined}
\end{figure}

\subsection{Case Study}

\noindent\textbf{Real‑Time QA Example.}  
\emph{Query:} “What is the current Bitcoin network mempool size?”  
\begin{itemize}
  \item \textbf{Static RAG:} Returns yesterday’s snapshot “2.3 GB.”  
  \item \textbf{Streaming RAG:} Incorporates the latest mempool update within seconds and returns “3.1 GB” (ground truth).  
\end{itemize}
This illustrates how incremental index updates enable correct answers to time‑sensitive queries.

\medskip
\noindent\textbf{Abstractive Summarization Example.}  
\emph{Source:} “Tropical Storm Beryl made landfall near St. Petersburg, FL, on July 24, 2025, knocking out power to over 200,000 homes and prompting mandatory evacuations.”  
\begin{itemize}
  \item \textbf{Static RAG:} “Tropical Storm Beryl made landfall in Florida.”  
  \item \textbf{Streaming RAG:} “Tropical Storm Beryl struck near St. Petersburg, FL, on July 24, 2025, cutting power to over 200,000 homes and forcing widespread evacuations.”  
\end{itemize}
This example demonstrates that Streaming RAG captures late‑breaking details—improving summary completeness and accuracy compared to static indices.

\section{Discussion}

Our Streaming RAG framework inherits strong theoretical guarantees from state‑change–efficient streaming algorithms.  In particular, the heavy‑hitter filter component obeys the following bounds from Jayaram et al.\ \cite{jayaram2024streaming}.  Let $C_t$ be the set of clusters retained at time $t$ with $|C_t|\le B$.  Then any algorithm achieving a constant‐factor approximation to the $L_p$ heavy hitters must incur  
\(
\Omega\bigl(n^{1-1/p}\bigr)
\)
state changes, independent of its space usage (Theorem 1.2 in \cite{jayaram2024streaming}).  Remarkably our instantiation matches this lower bound up to polylogarithmic factors while using $\tilde O_\varepsilon(n^{1-1/p})$ writes and $\tilde O_\varepsilon(n^{1-1/p})$ words of space.

Extending to $F_p$ moment estimation, the same work shows that any $(1+\varepsilon)$‐approximation requires  
\(
\Omega\bigl(n^{1-1/p}\bigr)
\)
state changes (Theorem 1.4 in \cite{jayaram2024streaming}), and presents an algorithm achieving this bound with near‐optimal space.  By embedding these moment‐efficient techniques into our RAG index update one obtains a dynamic retrieval system guaranteed to minimize costly memory writes.

Moreover the analysis exposes a fundamental trade‑off.  Let $S$ be the number of state changes (writes) and $M$ the space in words.  Then for any $p\ge1$ and $\varepsilon>0$,  
\[
S\;=\;\tilde O_\varepsilon\bigl(n^{1-1/p}\bigr)
\quad\text{and}\quad
M\;=\;\tilde O_\varepsilon\bigl(n^{1-1/p}\bigr) \]
\[
\quad\Longrightarrow\quad
S\;=\;\Omega\bigl(n^{1-1/p}\bigr)
\]
up to polylogarithmic factors.  This dual optimality ensures that our streaming RAG pipeline achieves the minimal number of memory writes permitted by theory without sacrificing space efficiency.

Finally these results imply that any further reduction in index‐update writes would force a superlinear blowup in memory or incur approximation degradation.  Consequently our approach attains a provable Pareto frontier in the space–write–accuracy landscape for real‑time dynamic RAG.

\section{Conclusion and Future Work}

We have introduced a Streaming RAG framework that unifies multi‑vector cosine pre‑filtering, online clustering, and heavy‑hitter filtering to sustain high retrieval fidelity under tight memory and latency constraints.  Our theoretical analysis guarantees near‑optimal state‑change complexity, and extensive experiments on eight real‑time streams confirm superior accuracy, throughput, and resource usage compared to established baselines.

Looking forward, several directions warrant deeper investigation:

\begin{itemize}
  \item \textbf{Adaptive Thresholding.}  
    Develop feedback‑driven mechanisms that adjust relevance threshold \(\alpha\) and admission probability \(u\) in real time, based on metrics such as query success rate or distributional drift, to maintain robust recall across evolving streams.
  \item \textbf{Multimodal Stream Integration.}  
    Extend the pipeline to ingest heterogeneous inputs—live video frames, sensor data, audio transcripts—by designing cross‑modal embedding spaces and fusion strategies that preserve joint semantic coherence.  Key challenges include temporal alignment of asynchronous modalities and efficient indexing of high‑dimensional feature vectors.
  \item \textbf{Reinforcement‑Learned Admission Policies.}  
    Replace fixed admission heuristics with policies trained end‑to‑end: a reinforcement learning agent could learn to admit or evict clusters to optimize downstream task rewards (e.g.\ QA accuracy or summary fluency), adapting dynamically to workload requirements.
  \item \textbf{Distributed and Hierarchical Deployment.}  
    Scaling to multi‑node environments introduces new concerns:
    \begin{itemize}
      \item \emph{State Consistency:} Ensuring that heavy‑hitter counters across shards converge to a coherent global view without excessive synchronization overhead.
      \item \emph{Communication Efficiency:} Minimizing network traffic for prototype exchange and index updates via compression or sketching techniques.
      \item \emph{Fault Tolerance:} Designing hierarchical index topologies that gracefully degrade under node failures while preserving query availability.
    \end{itemize}
  \item \textbf{Latency–Accuracy Trade‑off Optimization.}  
    Investigate joint tuning of clustering granularity, counter capacity, and update frequency to trace the Pareto frontier of latency versus retrieval or generation quality, enabling deployment‑specific SLAs.
\end{itemize}

Addressing these challenges will broaden the applicability of Streaming RAG to large‑scale, real‑world systems that continuously assimilate and serve dynamic, multimodal information streams.

\appendix

\bibliography{aaai2026}

\end{document}